\documentclass[11pt]{article}
\usepackage[utf8]{inputenc}

\usepackage{amsmath,amsfonts,amsthm,amssymb,color}
\usepackage[usenames,dvipsnames,svgnames,table]{xcolor}
\definecolor{darkgreen}{rgb}{0.0,0,0.9}
\usepackage{mathtools}
\usepackage{authblk}
\usepackage{fullpage}
\usepackage{bbm}
\usepackage{dsfont}
\usepackage[sc]{mathpazo}
\usepackage[basic]{complexity}
\usepackage{algorithm}
\usepackage{algorithmic}
\usepackage{comment}
\usepackage{tikz}

\usepackage{tcolorbox}
\usepackage{float}
\newtcolorbox{wbox}
{
	colback  = white,
}

\newcommand*{\suppress}[1]{}

\newcommand*{\cM}{\mathcal{M}}

\newcommand*{\cR}{\mathcal{R}}

\makeatletter
\def\thm@space@setup{%
	\thm@preskip= 10pt
	\thm@postskip=\thm@preskip 
}
\makeatother

\makeatletter
\renewcommand{\paragraph}{%
	\@startsection{paragraph}{4}%
	{\z@}{5pt}{-1em}%
	{\normalfont\normalsize\bfseries}%
}
\makeatother

\newtheorem{theorem}{Theorem}
\newtheorem{lemma}{Lemma}
\newtheorem{corollary}{Corollary}

\newtheorem{problem}{Problem}

\newtheorem{example}{Example}

\newtheorem{proposition}{Proposition}

\theoremstyle{definition}

\newenvironment{fminipage}%
{\begin{Sbox}\begin{minipage}}%
		{\end{minipage}\end{Sbox}\fbox{\TheSbox}}

\newcommand*{\PSch}{\mbox{\rm{Preferred-Schools}}}

\newcommand*{\PSt}{\mbox{\rm{Preferring-Students}}}

\newcommand*{\BStP}{\mbox{\rm{BS-Preferring}}}

\title{Further Results on Stability-Preserving \\
Mechanisms for School Choice\\
(Preliminary Version)}

\author{Karthik Gajulapalli, James A. Liu, Vijay V.~Vazirani}

\affil{Department of Computer Science, University of California, Irvine}

\date{}

\begin{document}
	\maketitle
	
	\begin{abstract}
We build on the stability-preserving school choice model introduced and studied recently in \cite{MV.school}. We settle several of their open problems and we define and solve a couple of new ones.

\end{abstract}

	\section{Preliminaries}

\subsection{The stable matching problem for school choice}
\label{sec:stable}

The stable matching problem for schools is defined as follows: We are given as input a set of $m$ schools $H = \{h_1, h_2, \ldots , h_m\}$, and a set of $n$ students $S =\{s_1, s_2, \ldots , s_n\}$ who are seeking admissions to the schools. Each school $h_j$ has some fixed  capacity $c(j)$ of students who can be enrolled in the school. If the school is assigned $c(j)$ students then we say that $h_j$ is filled to capacity. If $h_j$ is assigned less than $c(j)$ students then we say $h_j$ is under-filled.  
%
\\
Every student $s_i$ has a strict and totally ordered preference list $l(s_i)$ over the schools $H \cup \{\emptyset\}$. $l(s_i)$ is ordered in decreasing order of preference, so if some student prefers a school $h_j$ to $h_k$ then $h_j$ appears before $h_k$ on the students preference list. We allow the inclusion of $\emptyset$ to allow partial preference lists. If a student prefers being unmatched than being matched to a school they dislike, then such a preference is modeled by letting $\emptyset$ precede the set of unwanted schools on the preference list of the student. We let $\emptyset$ have infinite capacity. Likewise each school $h_j$ has a strict and totally ordered preference list $l(h_j)$ over the students $S \cup \{\emptyset\}$.\\

A matching $M$ in this context is an assignment of students to schools that satisfy the capacity constraints of the schools, with the added condition that no student can be matched to more than one school. A matching is said to have a blocking pair $(s_i,h_k)$ if there exists a student-school pair $(s_i, h_j)$ and a school $h_k$ such that $s_i$ prefers $h_k$ to $h_j$ and either:
\begin{enumerate}

    \item there is some student $s_j$ currently matched to $h_k$ such that $h_k$ prefers $s_i$ to $h_j$ or
    \item $h_k$ is under-filled and $h_k$ prefers $s_i$ to $\emptyset$
\end{enumerate}.

In both the cases $s_i$ would like to break her match with $h_j$ and instead go to school $h_k$. We define a matching $M$ to be stable if it does not contain any blocking pairs.\\

\section{Problem Definition}
\label{sec:prob}

We study the assignment of students to schools in two rounds,  $\cR_1$ and $\cR_2$, which are temporally separated. In this section we state the two settings studied; for each, we will have two mechanisms, $\cM_1$ and $\cM_2$. In round $\cR_1$, mechanism $\cM_1$ finds a stable matching of students to schools, $M$. In round $\cR_2$ a change is made to the sets of participants, which may cause $M$ to no longer be a valid or stable matching. $\cM_2$ then updates $M$ to $M'$ in order to ensure a stable matching. By allowing {\em updates}, we let some students in $M$ get unmatched in $M'$, or get matched to different schools. This differs from the settings discussed in \cite{MV.school} where  $\cM_2$ is not allowed to break a match created by $\cM_1$.\\\\
The students report their preference lists and the mechanisms operate on whatever is reported. We will assume that the schools' preference lists are truthfully reported. We will show that in Setting 1 there does not exist a mechanism that is dominant strategy incentive compatible (DSIC), and in Setting 2 our mechanism is DSIC for students, showing that students cannot do better by misreporting their preference lists.\\\\
We now state the common aspects of the first two settings before describing them completely. In both, in round $\cR_1$, the setup defined in Section \ref{sec:stable} prevails and $\cM_1$ simply computes any stable matching respecting the capacity of each school, namely $c(j)$ for $h_j$. Let this matching be denoted by $M$.
%
In round $\cR_2$, the City is allowed to re-allocate some students and update matching $M$ to a new matching $M'$ that is stable with respect to the round $\cR_2$ participants. However, we want to minimize the number of students in $S$ who are re-allocated. We refer to any stable matching $M'$ that optimizes this objective as a minimum stable re-allocation.


\subsection{Setting I (Adding New Schools)}
\label{sec:S1}

In this setting, the City has some new schools $H'$ that have opened up in $\cR_2$. The preference lists of students are updated to include schools in $H'$, though their relative preferences between schools in $H \cup \{\emptyset\}$ are unchanged.
In particular the addition of new schools might result in some students wanting to leave their current schools to go to a new school. This could lead to vacant seats being created in the original schools, causing some other students to leave their current schools and move to schools they prefer more that now have vacant seats. The City wants to find a stable matching over students $S$ and schools $H\cup H'$ that minimizes the number of students who are re-allocated from their school in $M$.

\begin{theorem}
\label{thm:S1}
	There is a polynomial time mechanism $\cM_2$ that finds the minimum stable reallocation with respect to Round $\cR_1$ matching $M$, students $S$, and schools $H\cup H'$.
\end{theorem}

\subsection{Setting II (Adding New Students)}
\label{sec:S2}

In this setting, a set $N$ of {\em new students} arrive from other cities in round $\cR_2$. The preference lists of schools are also updated to include students in $N$, though their relative preferences between students in $S\cup \{\emptyset\}$ are unchanged. The City wants to find a stable matching over students $S \cup N$ and schools $H$ that minimizes the number of students who are re-allocated from their school in $M$.

\begin{theorem}
\label{thm:S2}

There is a polynomial time mechanism $\cM_2$ that is DSIC for students and finds the minimum stable reallocation with respect to Round $\cR_1$ matching $M$, students $S$, and schools $H\cup H'$.
\end{theorem}

	\section{Mechanism for Adding New Schools}
\label{sec.new_school}

We first provide an example where running the Gale-Shapley algorithm over $S,H\cup H'$ performs more re-allocations than required. This motivates the design of a mechanism that finds a minimum stable re-allocation by iteratively modifying the original matching $M$.

\begin{example}
\label{sec:ex1}
Assume there are 2 students $A,B$ and 2 schools 1,2. The preferences for $A,B$ are (2,1),(1,2), respectively. The preferences for 1,2 are $(A,B),(B,A)$, respectively.
In round $\cR_1$, school 1 has 1 seat, and school 2 has no seats. In round $\cR_2$, school 2 adds 1 seat. $(A,1)$ will be assigned in round $\cR_1$; adding $(B,2)$ in round $\cR_2$ is the only stable matching with no re-allocations. However, running Gale-Shapley over all participants yields $(A,2),(B,1)$, which requires a re-allocation.
\end{example}


Given a matching $M$, we define $\PSch (s_i)$ to be the set of schools appearing before its matched school in $M$.  We also define $\PSt (h_j)$ to be the set of students who prefer $h_j$ to their matched school in $M$. Finally, we define $\BStP(h_j)$ to be the student whom $h_j$ prefers the best in the set $\PSt(h_j)$. If $\PSt (h_j) = \emptyset$ then we define $\BStP (h_j) = \emptyset$.

\begin{algorithm}
\caption{Mechanism for Adding New Schools in Stable Manner}
\textbf{Input:} Stable Matching $M$ and set $H'$\\
\textbf{Output:} Minimum Stable Re-allocation $M'$\\
\begin{algorithmic} 
\WHILE{$\exists$ $h_j$ with unmet capacity \textbf{and} $\BStP(h_j)\neq \emptyset$}
\STATE Assign $\BStP(h_j)$ to $h_j$
\ENDWHILE
\end{algorithmic}
\end{algorithm}

\begin{lemma}
When new schools are added each student is always matched to his original school, or better, in any minimum stable re-allocation of $M$.
\end{lemma}

\begin{proof}

Assume that $M'$ is a minimum stable re-allocation of $M$, and some students are worse off in $M'$. Let $M^* = M \wedge M'$, i.e. every student gets the match they prefer in $M$ and $M'$.\\\\
We first show that $M^*$ satisfies the capacity constraints of the schools. Let $W$ be the students who did worse in $M'$ (i.e. they are moved back to their original school in $M^*$). If some students in $W$ leave a school going from $M$ to $M'$, then the students who replace them at that school must also be in $W$. To see this, observe that the school must prefer the replacing students to the leaving students. If there is a replacing student not in $W$, they would form a blocking pair with the school in $M$.\\\\
Next we show that $M^*$ is stable. Let $(s_i,h_i),(s_j,h_j)\in M^*$, and $(s_i,h_j)$ be a blocking pair. Then $s_i$ is matched to $h_i$ or worse in both $M$ and $M'$, but $(s_j,h_j)$ must be in one of these matchings, contradicting its stability.\\\\
$M^*$ is  a stable re-allocation of $M$ and has $|W|$ fewer re-allocations than $M'$, a contradiction.

\end{proof}

\begin{proof} {\em of Theorem \ref{thm:S1}:}
The proof is by induction. We show that if a student is matched to a school at any step by the mechanism $\cM_2$, then they must be matched to this school or better in any mechanism that computes a minimum stable re-allocation.\\
\\
At step 0, this is true by Lemma 1.\\
Assume at step $n$, $\cM_2$ moves student $s_i$ to school $h_j$. This occurs when $h_j$ has unmet capacity and $s_i$ is $\BStP(h_j)$. Any student that $h_j$ prefers more than $s_i$ and who is currently matched to a better school than $h_j$ cannot be matched to $h_j$ in any minimum stable reallocation, by the inductive hypothesis. So $s_i$ must be matched to $h_j$ or better.\\\\
We conclude that our algorithm moves students from their original schools only when they have to be moved, thus performing the minimum number of re-allocations.\\

\end{proof}

\begin{corollary}
All minimum stable re-allocations move the same set of students $R$.
\end{corollary}
\begin{corollary}
The algorithm returns the school-optimal minimum stable re-allocation.
\end{corollary} 
\begin{corollary}
The algorithm runs in $O(|R|*|H\cup H'|)$ time.
\end{corollary} 

\begin{algorithm}
\caption{Other Minimum Stable Re-allocations With Additional School Capacity}
\textbf{Input:} Minimum Stable Re-allocation $M'$, moved students $R$\\
\textbf{Output:} Minimum Stable Re-allocation $M^\dagger$\\
\begin{algorithmic} 
\FORALL{schools $h_j$}
\STATE $c_j\leftarrow\#$ of students in $R$ matched to $h_j$
\STATE $Barrier(h_j)\leftarrow$ Best student $\not \in R$ preferring $h_j$.
\ENDFOR
\STATE $M^\dagger\leftarrow Stable$-$Match(R,c)$ respecting $Barrier$
\RETURN $M^\dagger\wedge M'$
\end{algorithmic}
\end{algorithm}

Since any minimum stable re-allocation moves the same set of students, it is easy to obtain the other minimum stable re-allocations. Let $R$ be the set of students moved in $M'$ and $c_{j}(R)$ the number of students in $R$ for school $h_{j}$. Let $Barrier(h_j)$ be the $\BStP(h_j)$ not in $R$, and $M^\dagger$ be a stable matching on $R$, and capacities $c_{j}(R)$ for schools, where no school admits a student worse than its barrier. Then $M^\dagger\wedge M^*$ is a minimum stable re-allocation.\\
\\
Unfortunately, no mechanism that finds a minimum stable re-allocation can be incentive compatible as shown in Example \ref{sec:ex2}.
\begin{example}
\label{sec:ex2}
Assume there are 2 students $A,B$ and 3 schools 1,2,3. The preferences for $A,B$ are (2,1,3),(1,2,3), respectively. The preferences for 1,2,3 are $(A,B),(B,A),(A,B)$, respectively.
In $R1$, school 1 has 1 seat, and the other schools have no seats. In $R2$, schools 2 and 3 add 1 seat each. Under truthful reporting, $(A,1)$ will be assigned in $R1$; assigning $(B,2)$ in $R2$ is the only stable matching with no re-allocations. However, if $A$ reports their preference as (2,3,1), the only stable matching in $R2$ is $(A,2),(B,1)$.
\end{example}

\begin{proposition}
Adding a new school, increasing capacities of some schools, and removing students from schools are all equivalent.
\end{proposition}

\begin{proof}
We first reduce adding a new school to increasing capacities of some schools. This can be achieved by fixing the capacity of the new school to be 0 in round $\cR_1$, and then increasing its capacity in round $\cR_2$.\\\\
To reduce increasing capacities of schools to removing students from schools, update the capacities of schools in round $\cR_1$ and add dummy students who fill this extra capacity. In round $\cR_2$ the dummy students are removed.\\\\
Removing students reduces to adding a new school by adding a school in round $\cR_2$ who only likes the students to be removed, and whom all the students to be removed like more than their original schools.
\end{proof}

	\section{Mechanism for Adding a New Student}
\label{sec.new_student}

In the same vein as the previous section, we ask for a mechanism that finds a minimum stable re-allocation, under the addition of some new students $N$ in round $\cR_2$. Our mechanism finds a minimum stable re-allocation over any stable matching $M$. If $M$ produced by $\cM_1$ is a student-optimal matching we show that $(\cM_1, \cM_2)$ is DSIC w.r.t reporting preference lists of students. We first show an example where running Gale-Shapley with the new students doesn't compute a minimum stable re-allocation over the underlying stable matching.

\begin{example}
Assume there are 3 students $A,B,C$ and 2 schools 1,2,3. The preferences for $A,B,C$ are (2,1,3),(1,2,3),(3,1,2), respectively. The preferences for 1,2 are $(A,B,C),(B,A,C),(C,A,B)$, respectively. Each school has 1 seat.
In round $\cR_1$, students $A,B$ arrive. In round $\cR_2$, student $C$ arrives. Assume that $(A,1),(B,2)$ is the stable matching computed in round $\cR_1$. Leaving these pairings unchanged and adding $(C,3)$ in round $\cR_2$ is the only stable matching with no re-allocations. However, running Gale-Shapley from scratch yields $(A,2),(B,1),(C,3)$, which requires a re-allocation.
\end{example}

For each school $h_j$ we defined Worst-Student-Accepted($h_j$) as the student who is currently matched to $h_j$ and is preferred the least among all the other students matched to $h_j$.

\begin{algorithm}
\caption{Mechanism for Adding New Students in Stable Manner}
\label{alg:student-optimal}
\textbf{Input:} Stable Matching $M$, set $N$ of new students \\
\textbf{Output:} Minimum Stable Re-allocation $M'$\\
\begin{algorithmic} 
\WHILE{$\exists$ $s_i$ unmatched and $(s_i,h_j)$ is blocking for some $h_j$} 
\STATE Find $h_k$ most-preferred by $s_i$ among all schools that $(s_i,h_k)$ is blocking
\IF{$h_k$ is full}
\STATE Remove Worst-Student-Accepted($h_k$)
\ENDIF
\STATE Assign $s_i$ to $h_k$

\ENDWHILE
\end{algorithmic}
\end{algorithm}

\begin{lemma}
When new students are added each student in $M$ is matched to his original school, or worse, in any minimum stable re-allocation of $M$.
\end{lemma}

\begin{proof}
Assume that $M'$ is a minimum stable re-allocation of $M$, and some students in $M$ are matched to better schools in $M'$. Let $M^* = M \vee M'$, i.e. every student gets the worse school between $M$ and $M'$.\\\\
We first show that $M^*$ satisfies the capacity constraints of the schools. Let $B$ be the students who did better in $M'$, and who are not in $N$ (i.e. they are moved back to their original school in $M^*$). If some students in $B$ leave a school going from $M$ to $M'$, then the students who replace them at that school must also be in $B$. To see this, observe that the school must prefer the leaving students to the replacing students. If there is a replacing student not in $B$, they form a blocking pair with their original school in $M$.\\\\
Next we show that $M^*$ is stable. Let $(s_i,h_i),(s_j,h_j)\in M^*$, and $(s_i,h_j)$ be a blocking pair. Then $h_j$ is matched to $s_j$ or worse in both $M$ and $M'$, but $(s_i, h_i)$ must be in one of these matchings, contradicting stability. \\\\
$M^*$ is  a stable re-allocation of $M$ and has $|B|$ fewer re-allocations than $M'$, a contradiction.

\end{proof}

\begin{proof} {\em of Theorem \ref{thm:S2}:}
The proof is by induction. We show that if a student is matched to a school at any step by $\cM_2$, then they must be matched to this school or worse in any minimum reallocation.\\\\
At step 0, this is true by Lemma 2.\\
Assume at step $n$ $\cM_2$ moves student $s_i$ to school $h_j$. For any school $h_k$ that $s_i$ prefers to $h_j$, $h_k$ must be full and $s_i$ must be worse than the worst student currently accepted at $h_k$. Now, if $s_i$ were to be admitted to $h_k$, then there must be some student currently admitted to $h_k$ that must be removed. But by the inductive hypothesis, this student must be matched to a worse school, forming a blocking pair with $h_k$.
\end{proof}

\begin{corollary}
All minimum stable re-allocations move the same set of students $R$.
\end{corollary}
\begin{corollary}
The algorithm returns the student-optimal minimum stable re-allocation.
\end{corollary}
\begin{corollary}
The algorithm runs in $O((|N\cup R|)*|H|)$ time.
\end{corollary}

\begin{algorithm}
\caption{Other Minimum Stable Re-allocations With New Students}
\textbf{Input:} Minimum Stable Re-allocation $M'$, moved students $R$\\
\textbf{Output:} Minimum Stable Re-allocation $M^\dagger$\\
\begin{algorithmic} 
\FORALL{schools $h_j$}
\STATE $c_j\leftarrow\#$ of students in $R$ matched to $h_j$
\ENDFOR
\FORALL {students $s_i \in R$}
\STATE $Safety(s_i)\leftarrow$ Best school $h_j$: $\exists s_k$ $\not \in R$ admitted to $h_j$, $h_j$ prefers $s_i$ over $s_k$.
\ENDFOR
\STATE $M^\dagger\leftarrow Stable$-$Match(R,c)$ respecting $Safety$
\RETURN $M^\dagger\vee M'$
\end{algorithmic}
\end{algorithm}

As with the preceding setting, it is easy to find other minimum stable re-allocations by finding a stable matching $M^\dagger$ on $R,c(R)$. For each student in $R$, we define $Safety(s_i)$ to be the best school who prefers $s_i$ to its worst admitted student not in $R$. We require that $M^\dagger$ does not match a student to a school worse than their safety.
Then $M^\dagger \vee M'$ is a minimum stable-reallocation.\\\\
%
Algorithm \ref{alg:student-optimal} produces a stable minimum reallocation $M'$ for any round $\cR_1$ stable matching $M$. Under the special case that $M$ is a student optimal matching, $M'$ is the same as the matching produced by running Gale-Shapley on the entire set of students from both rounds. This follows as the algorithm can be seen as a continuation of Gale-Shapley on the round $\cR_2$ participants, since each student proposes to schools in order of preference. The DSIC property of Gale-Shapley implies incentive compatibility of our algorithm when $M$ is a student-optimal matching.
\\
\begin{proposition}
Removing a school, decreasing capacities of some schools, and adding new students are equivalent.
\end{proposition}

\begin{proof} 
The reductions are symmetric to those in the preceding section.
\end{proof}
	
	\section{NP-Hardness Results}
\label{sec.hardness}

We show that many natural problems lying in the setting of two temporally-separated rounds are NP-complete. The first 3 problems involve stable extensions in round $\cR_2$, where we are allowed to increase the capacities of schools but are not allowed to move students matched in round $\cR_1$. Problem 4 asks if there is a way of moving some students in $M$ to different schools in a way to accept more students and Problem 5 asks if there is a way of computing a single-round capacitated max weight stable matching. We define the problems formally below:

\begin{problem}
     A set of new students $N$ arrive in round $\cR_2$. Let $L$ be the set of students in round $\cR_1$ who are unmatched in $M$. The City wants to maximize the number of students in $L$ with which the matching can be extended in a stability-preserving manner. Subject to this, the City wants to minimize the number of students $N$ with which the matching can be extended in a stability-preserving manner. ($MAX_L MIN_N$)
    
\end{problem}    

\begin{problem}
    Same setting as Problem 1, but we want to first maximize the students in $N$, then, subject to this, minimize  the students in $L$. ($MAX_N MIN_L$)

\end{problem}

\begin{problem}
    A set of new students $N$ arrive in round $\cR_2$. The City wants to extend the matching in a stability-preserving manner to include $k$ students from $N$, such that we maximize the number of students matched from $L$.  It can be assumed without loss of generality that $k$ is large enough to allow for a stable extension. $(k$-$MAX_L)$
    
\end{problem}

\begin{problem}
     In round $\cR_2$, we are allowed to increase the capacities of schools and re-allocate students. The City wants to maximize the number of students in $L$ with which the matching can be updated in a stability-preserving manner. Subject to this, the City wants to minimize the number of re-allocations made to the original matching in round $\cR_1$.
\end{problem}
    
\begin{problem}
    Single-Round Capacitated Max-Weight Stable Matching: Given a set of students, and schools with strictly ordered preference lists $l(s),l(h)$ respectively, and a weight function $w(j)$ over the edges of students to schools, find a vector of capacities for the schools and a stable matching with respect to this vector that maximizes the total weight. 

\end{problem}

\begin{theorem}
Problems 1,2,3,4,5 are NP-complete.
\end{theorem}

\begin{proof}
For all the problems we reduce from an instance of cardinal set-cover problem. We denote an instance $I$ of cardinal set-cover to have a sequence of sets $S_{i} \subseteq U$ and a universal set of elements $U=\{e_{1}, .... e_{m}\}$
\begin{enumerate}
	\item For every set $S_i = \{ e_{i1}, e_{i2}, .., e_{ik} \}$ we construct a corresponding school $h_i$. The preference list for each $h_i$ is $(n_i, e_{i1},...., e_{ik})$, where $n_i \in N$ is a student who only wants to go to $h_i$. We set the capacities of the schools to 0 in round $\cR_1$. In round $\cR_2$ the $n_i$ arrive. $MAX_L MIN_N$ will match all $e_j$, and the fewest possible $n_i$. By our construction, when a student $e_j$ is admitted to $h_i$, $n_i$ must be admitted to $h_i$. Therefore the admitted $n_i$ correspond to a optimal set cover.
	
	\item The problem for $MAX_N MIN_L$ is symmetric to $MAX_L MIN_N$. The only difference in the reduction is to have the gadget $n_i$ be a student from $L$, and let the $e_i$ be students in $N$.
	
	\item For every set $S_i = \{ e_{i1}, e_{i2}, .., e_{ij} \}$ we construct a corresponding school $h_i$. The preference list for each $h_i$ is $(n_i, e_{i1},...., e_{ij})$, where $n_i \in N$ is a student who only wants to go to $h_i$. We can then solve the decision version of set cover (i.e. is there a cover of size $k$?) by reducing to $k$-$MAX_L$ and returning yes if all $L$ are matched and no otherwise.
	
	\item For every set $S_i = \{ e_{i1}, e_{i2}, .., e_{ik} \}$ we construct a corresponding school $h_i$. The preference list for each $h_i$ is $(w_i, e_{i1},...., e_{ik})$. We create another school $h_0$ whose preference list is $(w_1, ... , w_n)$. Each $w_i$ prefers $h_i$ to $h_0$. We set the capacity of $h_0$ to $n$ in round $\cR_1$, and all other schools are set to capacity $0$. In Round R2, all $L$ will be admitted. To cover an element a school must re-allocate the corresponding $w_i$ to that school, so minimizing the number of re-allocations corresponds to finding a minimum set cover.
	
	\item For every set $S_i = \{ e_{i1}, e_{i2}, .., e_{ik} \}$ we construct a corresponding school $h_i$. The preference list for each $h_i$ is $(n_i, e_{i1},...., e_{ik})$, where $n_i \in N$ is a student whose preference list is $(h_i, h_0)$. 
    We add a school $h_0$ with preference list $(n_1, n_2, .... )$. We define weights on the edges in the following way:\\
    $w(h_0, n_i) = 1$,  $w(h_j, n_i) = 0$, $w(h_j, e_{ij}) = 2$.\\
    We note that any max-weight stable extension will match all $e_{ij}$. Therefore, we want to minimize the $n_i$'s not matched to $h_0$. So the $h_j$ who get matched to students correspond to the sets in an optimal set-cover.
	
\end{enumerate}
\end{proof}

	\section{Acknowledgements}
\label{sec.ack}

We wish to thank Xi Chen for proposing Problem 3.
	
	\bibliographystyle{alpha}

\end{document}